\documentclass{article}
\usepackage{amsmath,amssymb,amsthm}
\usepackage{authblk}
\usepackage{graphicx}
\usepackage{natbib}

\def\Prob{\mathrm{P}}
\def\E{\mathrm{E}}
\def\N{\mathbb{N}}
\def\R{\mathbb{R}}

\DeclareMathOperator*{\argmax}{argmax}
\DeclareMathOperator*{\argmin}{argmin}
\DeclareMathOperator*{\abs}{abs}
\usepackage[a4paper,left=3cm,right=3cm,top=3cm,bottom=3cm]{geometry}

\newtheorem{theorem}{Theorem}
\newtheorem{lemma}[theorem]{Lemma}
\theoremstyle{definition}
\newtheorem{definition}[theorem]{Definition}

\author[*]{Patrick Rubin-Delanchy}
\author[**]{Carey E. Priebe}
\author[**]{Minh Tang}
\affil[*]{University of Oxford and Heilbronn Institute for Mathematical Research, U.K.}
\affil[**]{Johns Hopkins University, U.S.A.}
\date{}
\title{Consistency of adjacency spectral embedding for the mixed membership stochastic blockmodel}

\begin{document}
\maketitle
\begin{abstract}
The mixed membership stochastic blockmodel is a statistical model for a graph, which extends the stochastic blockmodel by allowing every node to randomly choose a different community each time a decision of whether to form an edge is made. Whereas spectral analysis for the stochastic blockmodel is increasingly well established, theory for the  mixed membership case is considerably less developed. Here we show that adjacency spectral embedding into $\R^k$, followed by fitting the minimum volume enclosing convex $k$-polytope to the $k-1$ principal components, leads to a consistent estimate of a $k$-community mixed membership stochastic blockmodel. The key is to identify a direct correspondence between the mixed membership stochastic blockmodel and the random dot product graph, which greatly facilitates theoretical analysis. Specifically, a $2 \rightarrow \infty$ norm and central limit theorem for the random dot product graph are exploited to respectively show consistency and partially correct the bias of the procedure.
\end{abstract}
\section{Introduction}
Network modelling is a thriving area of statistics, as is evident from the number of recent publications in top-tier statistical journals, e.g. \citet{lei2015,gao2015rate,lei2016goodness,klopp2017oracle,caron2014sparse} (the last a Royal Statistical Society discussion paper, 2017) and machine-learning, e.g. \citet{addario2015exceptional,wang2016trend,ho2016latent}, as data with this type of structure are increasingly found in all areas of science and beyond \citep{barabasi2016network}. 

The stochastic blockmodel \citep{holland1983stochastic}, in particular, where two nodes have an edge with probability dependent only on the communities that each belong to, has proved very popular in practice \citep{karrer2011stochastic} and a large body of statistical theory for this model is emerging \citep{rohe2011spectral,choi2012stochastic,lei2015,lei2016goodness}. An important result, greatly influencing the research of this paper, is the proof that a previously very popular data analysis technique called spectral clustering (spectral embedding followed by $k$-means clustering) \citep{von2007tutorial} provides a consistent estimate of the stochastic blockmodel, as shown by \citet{rohe2011spectral} and \citet{lei2015}. 

The mixed membership stochastic blockmodel \citep{airoldi2008mixed} is a natural extension of the stochastic blockmodel whereby each node can belong to a number of different communities, and is almost equally as popular, at least judging by the number of citations (over 1100 at the time of writing). However, here it seems that formal statistical theory lags behind --- in particular the connection between this model and spectral embedding is less clear. Yet, because of the afore-mentioned results for the stochastic blockmodel, it is natural to expect that a spectral approach could be profitable, and this leads us to the subject of this paper. 

Our main result is that, under reasonable conditions, \emph{spectral embedding of the adjacency matrix} \citep{athreya2016limit}, followed by fitting the \emph{minimum volume enclosing convex $k$-polytope} \citep{lin2016fast} to the $k-1$ principal components, provides a consistent estimate for the undirected mixed membership stochastic blockmodel. 






A number of other estimation techniques for this model have been developed, ranging from the variational inference schemes proposed by the original authors \citep{airoldi2008mixed,gopalan2013efficient}, tensor approaches to subgraph counts \citep{anandkumar2014tensor}, full Markov Chain Monte Carlo methods (available in online code repositories). Spectral techniques for a number of slightly modified models, with different consistency results, are available in a few other (apparently yet unpublished) papers \citep{zhang2014detecting,mao2016provable}. One advantage of our approach is that it is very simple, for example, requiring about three lines of code to estimate $B$ (excluding implementations of the eigendecomposition and fitting the polytope).

A key insight of this paper is recognising a direct connection between the mixed membership stochastic blockmodel and the \emph{random dot product graph} \citep{nickel06,young2007random,athreya2016limit}. This has allowed us to present a relatively developed theoretical analysis at very little cost: first, from a $2 \rightarrow \infty$ norm result given in \citet{lyzinski2017community} it is straightforward to show that our proposed procedure is consistent. Second, using a central limit theorem for spectral embedding given by \citet{athreya2016limit}, we are able to quantify and partly correct its finite-sample bias.

\section{Mixed membership stochastic blockmodels as random dot product graphs}\label{sec:mmsbm_as_rdpg}
The focus of this article is on modelling a random, undirected, simple graph with no self-loops on $n \in \N$ nodes. Any graph considered has its nodes labelled $1, \ldots, n$ and is identified via its \emph{adjacency matrix}, which is a symmetric, hollow matrix $A \in \{0,1\}^{n \times n}$ where $A_{ij}=1$ when the nodes $i$ and $j$ have an edge. 

\citet{airoldi2008mixed} introduce a \emph{mixed membership stochastic blockmodel}, which is here modified so as to generate undirected graphs.
\begin{definition}[Mixed membership stochastic blockmodel --- undirected version] 
Let $k \in \N$, $B \in [0,1]^{k \times k}$ and symmetric, and $\alpha \in \R_+^k$. We say that $(\Pi, A) \sim \text{MMSBM}(B,\alpha)$ if the following hold. First, let $\pi_1, \ldots, \pi_n \overset{i.i.d.} \sim \text{Dirichlet}(\alpha)$ and define $\Pi = [\pi_1, \ldots, \pi_n]^T \in \left(\boldsymbol{\Delta}^{k-1}\right)^n \subset [0,1]^{n \times k}$, where $\boldsymbol{\Delta}^m$ denotes the standard $m$-simplex.
Second, the matrix $A \in \{0,1\}^{n \times n}$ is defined to be a symmetric, hollow matrix such that for all $i < j$, \emph{conditional on $\Pi$},
\[A_{ij} \overset{ind}{\sim} \text{Bernoulli}\left(B_{z_{i\rightarrow j}, z_{j\rightarrow i}}\right),\]
where 
\[z_{i\rightarrow j} \overset{ind}{\sim} \text{multinomial}(\pi_i) \quad \text{and}\quad  z_{j\rightarrow i} \overset{ind}{\sim} \text{multinomial}(\pi_j).\]
\end{definition}

In a completely separate sequence of papers, \citet{nickel06,young2007random,athreya2016limit} introduce and analyse the following \emph{random dot product graph} model:
\begin{definition}[Random dot product graph]\label{def:RDPG}
Let $F$ be a distribution on a convex set $\mathcal{X} \subset \R^d$, such that $x^T x' \in [0,1]$ for all $x,x' \in \mathcal{X}$. We say that $(X,A) \sim \text{RDPG}(F)$ if the following hold. First, let $X_1, \ldots, X_n \overset{i.i.d} \sim F$ and define:
\[X = [X_1, \ldots, X_n]^T \in \mathcal{X}^{n \times d} \quad \text{and} \quad P = X X^T.\]
Second, the matrix $A \in \{0,1\}^{n \times n}$ is defined to be a symmetric, hollow matrix such that for all $i < j$, \emph{conditional on $P$},
\[A_{ij} \overset{ind}{\sim} \text{Bernoulli}\left(P_{ij}\right).\]
\end{definition}
As mentioned in the introduction, a key new notion of this paper is the connection between the two models, now made explicit.

\begin{lemma}[Mixed membership stochastic blockmodels are random dot product graphs]\label{lem:mmsbm_rdpg}
Let $(\Pi, A) \sim \text{MMSBM}(B,\alpha)$, and assume that $B$ is non-negative definite, with eigendecomposition $B=U \Sigma U^T$. Let $X=\Pi U \Sigma^{1/2}$. Then, $(X,A) \sim \text{RDPG}(F)$, for some $F$, with support the convex hull of the columns of $\Sigma^{1/2} U^T$.
\end{lemma}

\begin{proof}
The rows of $X$, transposed into vectors, are given by $X_i = \Sigma^{1/2} U^T \pi_i$, $i = 1, \ldots, n$. Therefore, $X_1, \ldots, X_n$ are i.i.d. from a distribution $F$, where $F$ is the distribution of a random vector $\Sigma^{1/2} U^T \pi$, with $\pi \sim \text{Dirichlet}(\alpha)$, which has support the convex hull of the columns of $\Sigma^{1/2} U^T$. Now,
\begin{align*}
\Prob(A_{ij} =1 \mid \Pi) &= \sum_{u=1}^{k} \sum_{v=1}^{k} \Prob(A_{ij} =1 \mid z_{i\rightarrow j}=u, z_{j\rightarrow i} = v) \Prob(z_{i\rightarrow j}=u\mid \Pi) \Prob(z_{j\rightarrow i} = v \mid \Pi), \\
& = \sum_{u=1}^{k} \sum_{v=1}^{k} B_{uv} \Pi_{iu} \Pi_{jv},
\end{align*}
so that together with conditional independence we have


\[A_{ij} \mid \Pi \overset{ind}\sim \text{Bernoulli}(\pi_i^T B \pi_j).\]
But 
\[\pi_i^T B \pi_j = (\Sigma^{1/2} U^T \pi_i)^T\Sigma^{1/2} U^T \pi_j = \left(X X^T\right)_{ij},\]
and therefore 
\[A_{ij} \mid \Pi \overset{ind}\sim \text{Bernoulli}\left\{\left(X X^T\right)_{ij}\right\} \Rightarrow A_{ij} \mid X \overset{ind}\sim \text{Bernoulli}\left\{\left(X X^T\right)_{ij}\right\}.\]
We recognise the generative model of a random dot product graph with parameter $F$.
\end{proof}



\section{Spectral estimation for the mixed membership stochastic blockmodel}
\label{sec:spectral_estimation}
The statistical problem we now consider is: given a single observation of $A$, estimate $B$ and $\alpha$, assuming $k$ and $d=\mathrm{rank}(B)$ are known. Importantly, we continue to treat $\Pi$ as \emph{random} (and unknown), whereas $B$ and $\alpha$ are \emph{fixed} (and unknown).

Additionally, motivated by Lemma~\ref{lem:mmsbm_rdpg}, we assume that $B$ is non-negative definite, so that $A$ can equally well be considered to have been generated from the model $\text{MMSBM}(B,\alpha)$ or a model $\text{RDPG}(F)$. The distribution $F$ is hereafter assumed to have the form induced by the mixed membership stochastic blockmodel, i.e.,  $F$ is the distribution of a random vector $\Sigma^{1/2} U^T \pi$, where $\pi \sim \text{Dirichlet}(\alpha)$.

Let $\mathcal S$ denote the support of $F$, that is, the convex hull of the columns of $\Sigma^{1/2} U^T$. Then $\mathcal S$ is a convex polytope with dimension $d-1$, and $l$ distinct vertices, denoted $v_1, \ldots, v_l$, where $d \leq l \leq k$. We use ``convex polytope'' rather than ``simplex'' to stress that $\mathcal{S}$ may be less than $(k-1)$-dimensional (if $d < k$).

As mentioned in the introduction, the estimation problem will be tackled using spectral embedding of the adjacency matrix:
\begin{definition}[Adjacency spectral embedding]
Given an adjacency matrix $A \in \{0,1\}^{n \times n}$, an \emph{Adjacency Spectral Embedding} (ASE) of $A$ into $\R^d$ is a matrix 
$\hat X = [\hat X_1, \ldots, \hat X_n]^T = U_A \{\abs(S_A)\}^{1/2} \in \R^{n \times d}$, where $S_A$ denotes the diagonal matrix containing, in decreasing order, the $d$ largest eigenvalues of $A$ by magnitude and $U_A$
 is a matrix containing corresponding orthonormal eigenvectors of $A$ in its columns. 
\end{definition}

Figure~\ref{fig:pipeline}a) shows the ASE into $\R^3$ of a simulated realisation of the mixed membership stochastic blockmodel, with $n=5000$, $k=3$, $\alpha = (1,1,1)$ and
\begin{equation}
B = \left[\begin{array}{ccc}0.9 & 0.2 & 0.3\\ 0.2 & 0.9 & 0.5 \\ 0.3 & 0.5 & 0.9\end{array}\right].\label{eq:B}
\end{equation}

\begin{figure}
\centering
\includegraphics[width=13cm]{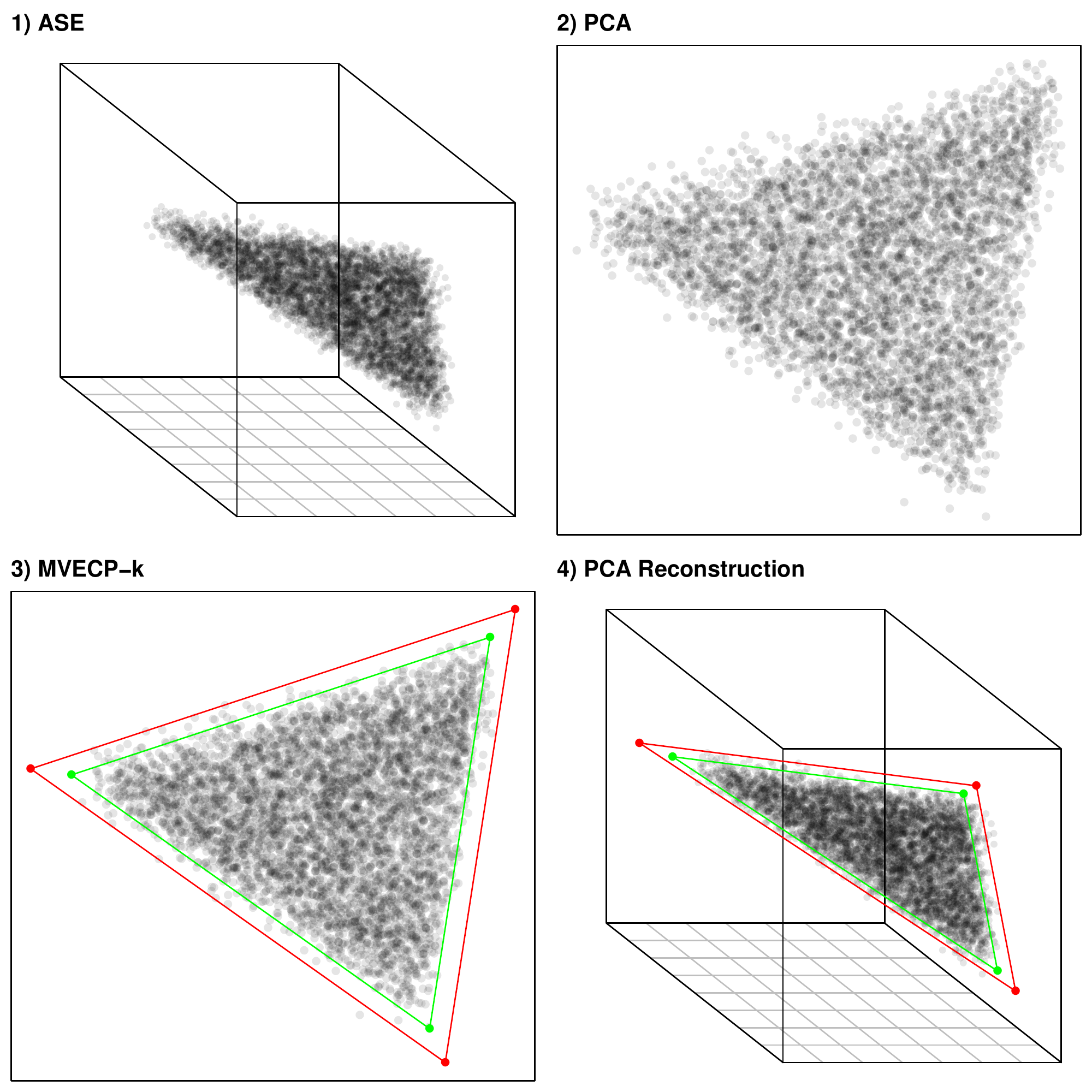}
\caption{Estimation pipeline. 1) Construct an ASE of the graph into $d$ dimensions. 2) Obtain the $d-1$ principal components of the data. 3) Compute the MVECP-$k$ around the data (red) and, optionally, shrink the polytope by a factor proportional to $n^{-1/2} \log^{1/2}(n)$ (green). 4) Apply the inverse PCA rotation and translation to recover the polytope in $d$-dimensional space. The vertices of the polytope provide an estimate for $B$.}\label{fig:pipeline}
\end{figure}
Intuitively, the figure suggests an obvious `estimation pipeline': First, construct the ASE into $\R^d$ (note that often, like in this example, $d=k$, but in general $d \leq k$). Second, find the principal $(d-1)$-dimensional hyperplane, using principal component analysis (PCA), and project the points onto the plane, as shown in Figure~\ref{fig:pipeline}b). Third, find the \emph{minimum volume enclosing convex $k$-polytope} (MVECP-$k$) around the projected points, as shown in Figure~\ref{fig:pipeline}c) (in red). Due to noise, this polytope can be `too large'; in Section~\ref{sec:shrinking} we will propose a correction (in green). Finally, reconstruct the polytope in $\R^d$ to obtain an estimate of $B$. More formally, we will prove consistency of the following procedure:


\begin{definition}[Spectral estimation of the mixed membership stochastic blockmodel]\label{def:spectral_estimation}
Assuming $k$ and $d$ are known,
\begin{enumerate}
\item Let $\tilde X = [\tilde X_1, \ldots, \tilde X_n]^T$ denote the $d-1$ principal components of $\hat X$, that is, $\tilde X = (\hat X - \hat M) U_{\hat C}$ where $\hat M = [n^{-1} \sum_{i=1}^n \hat X_{i}, \ldots, n^{-1} \sum_{i=1}^n \hat X_{i}]^T \in \R^{n \times d}$,
$\hat C = n^{-1}(\hat X - \hat M)^T (\hat X - \hat M)$, and $U_{\hat C}$ is a matrix containing $d-1$ orthonormal eigenvectors corresponding to the $d-1$ largest eigenvalues of $\hat C$. 
\item Let $\tilde{\mathcal{S}}$ denote the MVECP-$k$ around $\tilde X_1, \ldots, \tilde X_n$, and $\hat{\mathcal{S}} = \{U_{\hat C} x + n^{-1} \sum_{i=1}^n \hat X_{i}: x \in \tilde{\mathcal{S}}\}$ the reconstructed convex $k$-polytope in $\R^d$.  
\item Let $\hat B = \hat V \hat V^T$, where $\hat V = [\hat V_1, \ldots, \hat V_k]^T$ and $\hat V_1, \ldots, \hat V_k$ are the vertices of $\hat{\mathcal S}$. 
\item For $i \in 1, \ldots, n$, let $\hat \pi_i \in \boldsymbol{\Delta}^{k-1}$ be a vector satisfying $\hat X_i = \sum_{j=1}^k \hat \pi_{ij} \hat V_j$ and finally
\item Let $\hat \alpha$ denote the conditional maximum likelihood estimate
\[\hat \alpha = \underset{\alpha \in \R_+^k} \argmax \: \text{B}(\alpha)^{-1}\prod_{i=1}^n \prod_{j=1}^k \hat \pi_{ij}^{\alpha_j -1},\]
where \[\mathrm{B}(\alpha) = \frac{\prod_{j=1}^k \Gamma(\alpha_j)}{\Gamma\left(\sum_{j=1}^k \alpha_j\right)}.\]
\end{enumerate}
\end{definition}


\begin{figure}
\centering
\includegraphics[width=13cm]{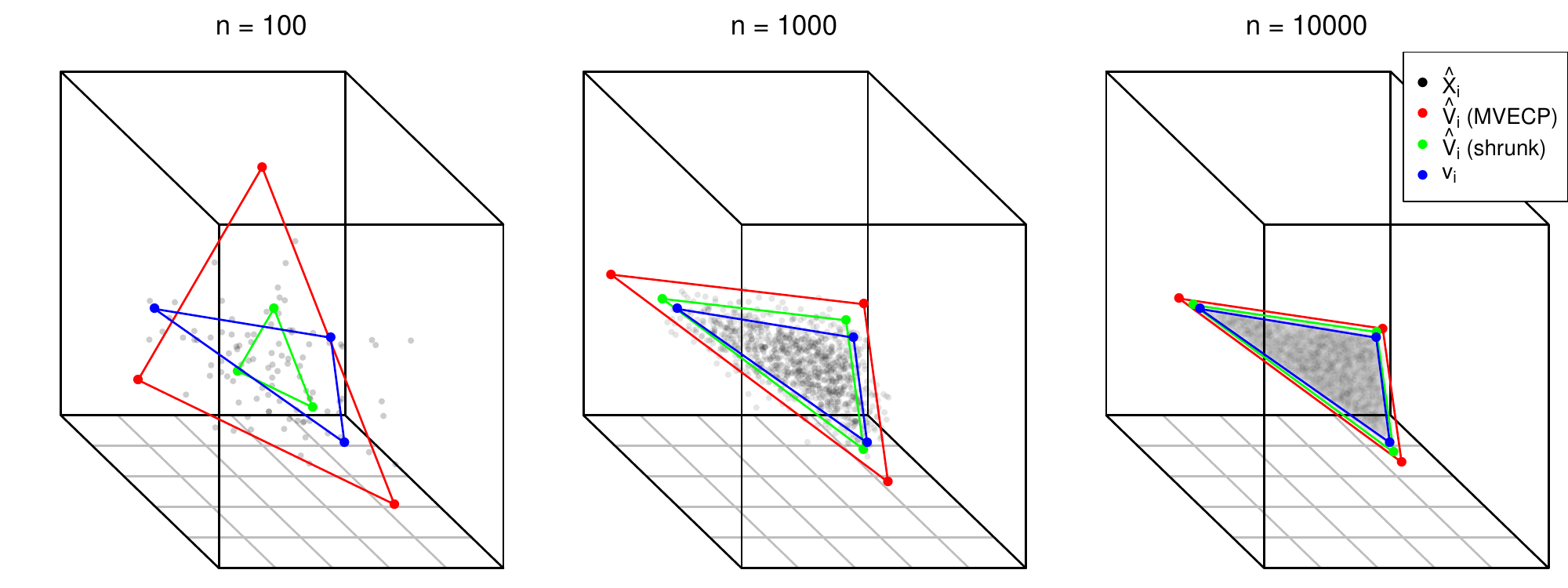}
\caption{ASE of a random realisation of the mixed membership stochastic blockmodel, for three different values of $n$, where $B$ is $3 \times 3$ with full rank (given explicitly in main text), and $\alpha = (1,1,1)$. The grey point cloud in $\R^3$ shows the estimated latent positions $\hat X_1, \ldots, \hat X_n$, the MVECP-3 is shown in red, the same polytope shrunk by a factor proportional to $n^{-1/2} \log^{1/2}(n)$ is shown in green, and the true vertices of $\mathcal S$ are shown in blue. Further details in main text.
}\label{fig:MVES_figure_ns}
\end{figure}

\begin{figure}
\centering
\includegraphics[width=13cm]{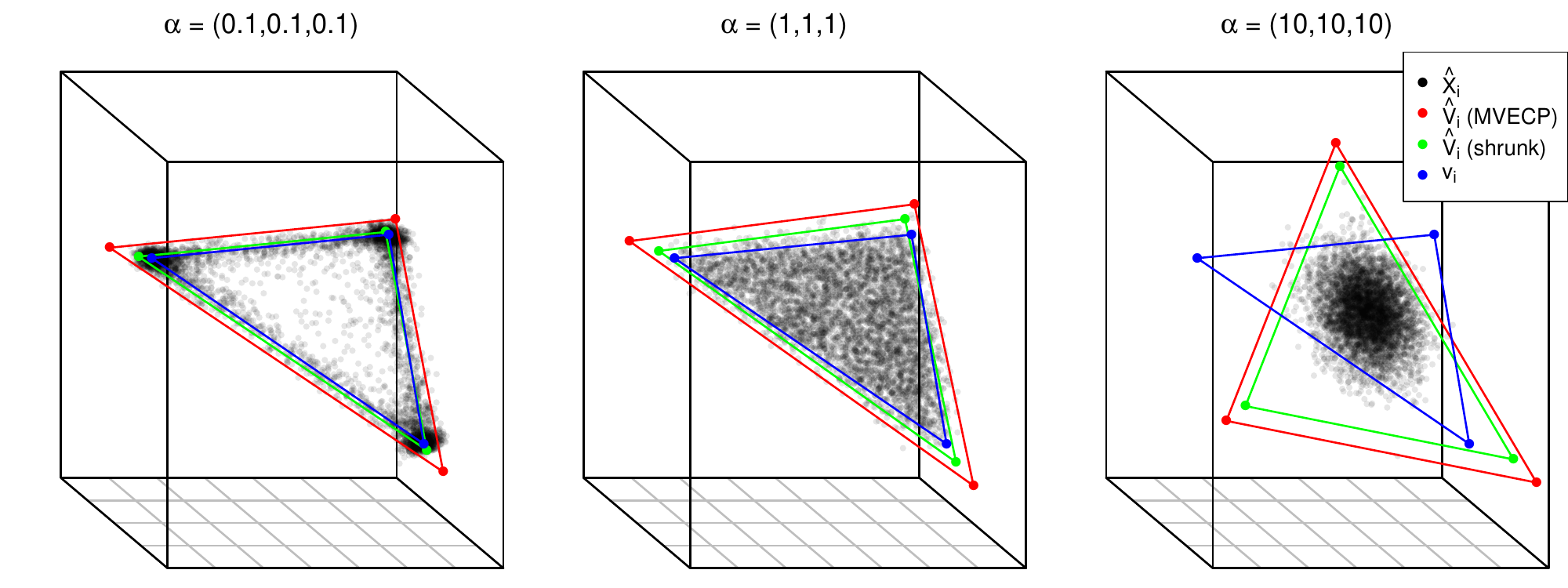}
\caption{ASE of a random realisation of the mixed membership stochastic blockmodel, for three different values of $\alpha$, where $B$ is $3 \times 3$ with full rank (given explicitly in main text), and $n=5000$. The grey point cloud in $\R^3$ shows the estimated latent positions $\hat X_1, \ldots, \hat X_n$, the MVECP-3 is shown in red, the same polytope shrunk by a factor proportional to $n^{-1/2} \log^{1/2}(n)$ is shown in green, and the true vertices of $\mathcal S$ are shown in blue. Further details in main text.}\label{fig:MVES_figure_alphas}
\end{figure}

While computing the MVECP-$k$ is in theory NP-hard \citep{packer2002np}, there exist efficient approximate solutions in practice. We have found the hyperplane-based algorithm by \citet{lin2016fast} fit for purpose, with a standard computer taking under a second to return an approximate MVECP-$k$ given, e.g., $n=10,000$ points in $\R^3$. We will say MVECP-$k$ to mean either the actual optimum or the practical approximation: it will always be clear from the context which is meant.

Figures~\ref{fig:MVES_figure_ns} and \ref{fig:MVES_figure_alphas} show how different quality spectral estimates are obtained depending on the values of $n$ and $\alpha$. A more quantitative analysis is presented later (Figure~\ref{eq:boxplot_error}). In each plot, the grey points are the estimated latent positions $\hat X_1, \ldots, \hat X_n$. The MVECP-$k$ and shrunk MVECP-$k$ are shown in red and green, as before, but now the true polytope, $\mathcal{S}$, is also shown, in blue. The parameter vector $\alpha$ controls where the true $X_i$ concentrate within $\mathcal{S}$. In particular, when all coordinates are commensurate and high (respectively, low) the mass concentrates towards the centre (respectively, the edges). As could have been expected, estimates improve as $n$ increases and each coordinate of $\alpha$ decreases.

To summarize our main result, 
our proposed estimates for $B$, $\pi_1, \ldots, \pi_n$ and $\alpha$ are consistent if $B$ is positive definite ($k=d$). To state the result in full, including allowing $B$ to be non-negative definite ($d \leq k$), let $\lVert \cdot \rVert$, $\lVert \cdot \rVert_{\mathrm{F}}$ denote the Euclidean and Frobenius norms respectively; let $\ominus$ denote the symmetric difference of sets and $\lambda$ Lebesgue measure; finally, let $\mathrm{O}(d)$ denote the orthogonal group $\{W \in \R^{d \times d}: W^T W = W W^T = I\}$ where $I$ is the identity matrix, $S_k$ the symmetric group on $k$ letters, and $M^{(\rho)}$, $v^{(\rho)}$ the row-column permutation (respectively, direct permutation) of a $k \times k$ matrix $M$ (respectively, a $k$-dimensional vector $v$) by $\rho \in S_k$, e.g. $M^{(\rho)}_{ij} = M_{\rho(i)\rho(j)}$. Then:
\begin{theorem}[Consistency of spectral estimates]\label{thm:main_theorem}
$\hat{\mathcal{S}}$ is a consistent estimator for $\mathcal{S}$, up to identifiability constraints, in the sense that, for any $\delta > 0$, 
\begin{equation}
\underset{n \rightarrow \infty} {\lim} \Prob\left[\underset{W \in \mathrm{O}(d)}{\min}\{\lambda(W \hat{\mathcal{S}} \ominus \mathcal S)\} \leq \delta\right] = 1. \label{eq:main_theorem}
\end{equation}
If $k = l$, then $\hat B$ is a consistent estimator for $B$, up to row-column permutations, i.e., for any $\delta >0$,
\[\lim_{ n\rightarrow \infty} \Prob\left\{\underset{\rho \in S_k}{\min}(\lVert \hat B^{(\rho)} - B\rVert_{\mathrm{F}}) \leq \delta\right\} = 1. \]
If additionally $d=l=k$ ($B$ is positive definite), then $\hat{\pi}_1, \ldots, \hat{\pi}_n$ and $\hat \alpha$ are consistent estimators of $\pi_1, \ldots, \pi_n$ and $\alpha$ respectively, up to permutation, i.e. for any $\delta >0$,
\begin{equation}
\lim_{ n\rightarrow \infty} \Prob \left[\underset{\rho \in S_k}{\min}\left\{ \max \lVert \hat{\pi}^{(\rho)}_i - \pi_i \rVert\right\} \leq \delta\right] = 
\lim_{ n\rightarrow \infty} \Prob\left\{\underset{\rho \in S_k}{\min}\left(\lVert \hat \alpha^{(\rho)} - \alpha \rVert \leq \delta\right)\right\} = 1. \label{eq:pi_estimate}
\end{equation}

\end{theorem}
The proof is relegated to the appendix as a number of technicalities distract from the main mathematical point: that because the MVECP-$k$ encloses all points by definition, Theorem~\ref{thm:main_theorem} is evidently only possible if the maximum deviation of any $\hat X_i$ to its true value $X_i$ can be controlled, probabilistically. The result we need is given by \citet{lyzinski2017community}. Let $E_n$ denote the event that 
\begin{equation}
\underset{i \in \{1, \ldots, n\}}{\max} \: \lVert W \hat X_i - X_i \rVert \leq d_n, \label{eq:bound}
\end{equation}
for some $W \in \mathrm{O}(d)$, where
 \[d_n=\frac{c d^{1/2}\log^2(n)}{\sqrt{n}} \rightarrow 0,\] 
and $c$ is some fixed constant. Then, $\Prob\left(E_n\right) \rightarrow 1$. 

\section{Shrinking the polytope}
\label{sec:shrinking}
As Figures~\ref{fig:MVES_figure_ns} and \ref{fig:MVES_figure_alphas} illustrate, the MVECP-$k$ has a tendency to be too large since, obviously, it is susceptible to outliers. In this section, we derive a rate (as a function of $n$) by which the polytope can be shrunk, towards the centre of the point cloud, for better results. For consistency to be preserved, it is necessary and sufficient that this rate tends to zero. Although our recommendation is the same when $l<k$, to simplify the following discussion  it is assumed that $k=l$ (a common scenario).

Again we appeal to the theory of random dot product graphs, this time a central limit theorem by \citet{athreya2016limit}: there exists a sequence of orthogonal matrices $W_n \in O(d)$ such that for any $i \in \{1, \ldots, n\}$ and $x \in \mathcal{X}$,
\[\mathcal{L}\left\{n^{-1/2} (W_n \hat X_i - X_i)\mid (X_i=x)\right\} \rightarrow \text{Normal}\{0, \Psi(x)\},\]
where the covariance $\Psi(x)$ is fixed as a function of $x$. Furthermore, at fixed indices $i_1, \ldots, i_m$ and fixed points $x_j \in \mathcal{X}, j = i_1, \ldots, i_m$, the random vectors $n^{-1/2} (W_n \hat X_j - X_j)\mid (X_j=x_j), j = i_1, \ldots, i_m$ are asymptotically independent.

We use the result that $\max(Z_j)/\log(m) \overset{a.s.} \rightarrow c$, for some constant $c > 0$ if $Z_1, \ldots, Z_m$ are i.i.d. Chi-square random variables \citep[Example 3.5.6]{embrechts2013modelling}. 
Choosing $x_j = x, j = i_1, \ldots, i_m$, for any $x \in \mathcal{X}$, this implies that the maximal coordinate of $n^{-1/2} (W_n \hat X_j - x)\mid (X_j=x), j = i_1, \ldots, i_m$, along any direction, decreases as $n^{-1/2}\log^{1/2}(n)$ (the marginals along this direction are independent zero-mean Gaussian random variables, and so their squares are independent Chi-square random variables up to rescaling). Heuristically, we might then assume that the error in the MVECP-$k$ vertices relative to the vertices $v_1, \ldots, v_k$ of $\mathcal S$ decrease as $n^{-1/2}\log^{1/2}(n)$ too, by considering in turn subsets of $\{X_i\}$ that fall close to each  $v_i$ (i.e. letting $x = v_i$ in the above argument).

\begin{figure}
\centering
\includegraphics[width=13cm]{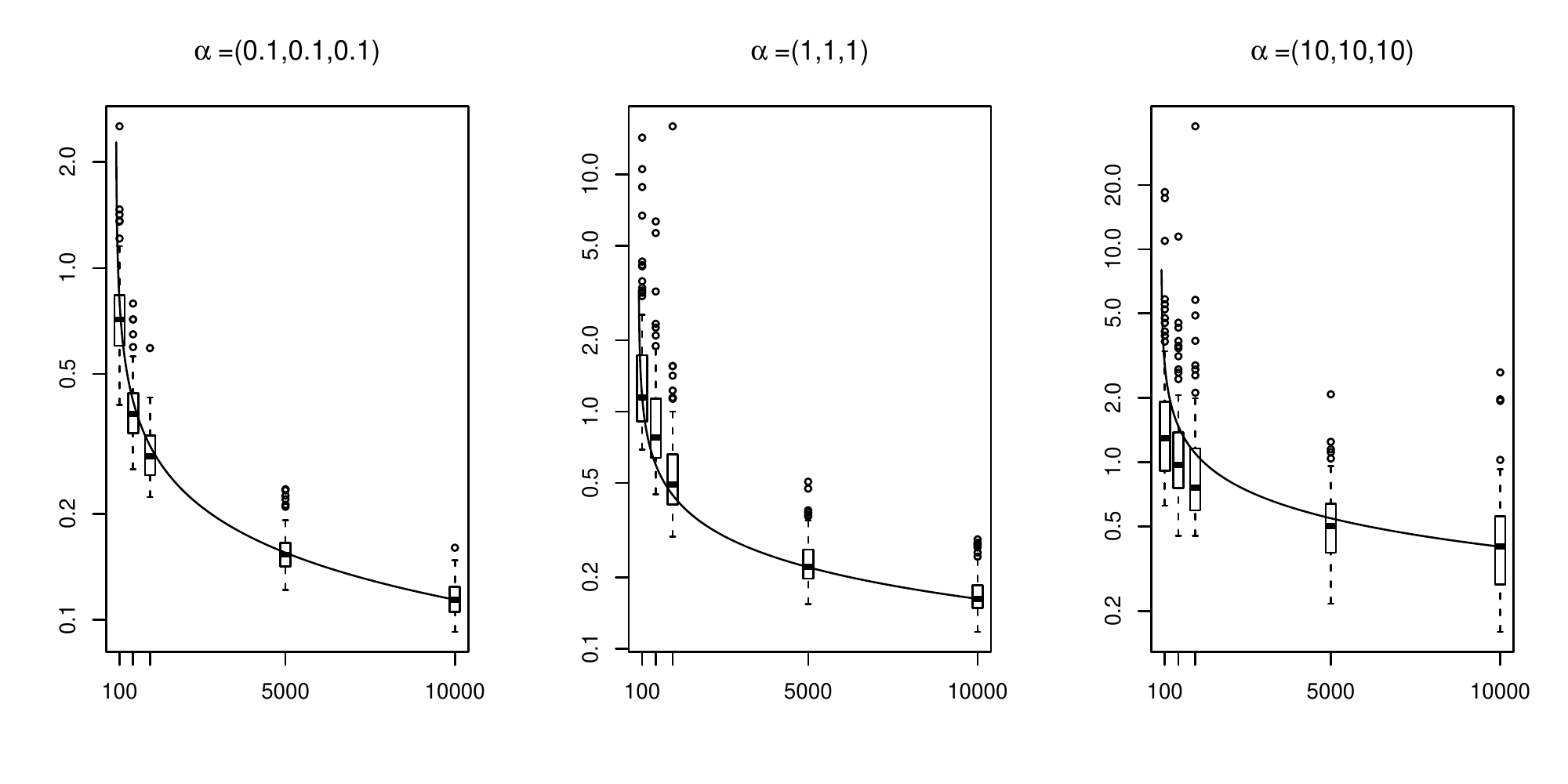}
\caption{Boxplots of the maximum distance, for $n = 100, 500, 1000, 5,000, 10,000$, between a vertex of the MVECP-$k$ $\hat{\mathcal{S}}$ and the corresponding true vertex of $\mathcal S$, for three different values of $\alpha$, where $B$ is $3 \times 3$ with full rank (given explicitly in main text), based on 100 simulations for each $n$. The black curves are proportional to $n^{-1/2}\log^{1/2}(n)$, with the constant chosen in each of the three panels so that the curve and the sample median meet at $n = 10,000$. The y-axis is on the log-scale. Further details in main text.}\label{fig:vertex_distance}
\end{figure}
Figure~\ref{fig:vertex_distance} shows samples of the simulated maximum distance between any of the vertices of the MVECP-$k$ $\hat{\mathcal{S}}$ and the corresponding vertex of $\mathcal S$, for different values of $n$ and $\alpha$, and $B$ given as usual by \eqref{eq:B}. The black line shows the rate $n^{-1/2}\log^{1/2}(n)$, which is in surprisingly good agreement with the asymptotic theory. 

A simple plausible avenue for improvement is therefore to shrink the MVECP-$k$ towards the centre of the point cloud by a rate $a n^{-1/2}\log^{1/2}(n)$, for some constant $a\leq 0$. By comparison, for their (entirely different) application, \citet{lin2016fast} suggest shrinking the MVECP-$k$ of points corrupted by noise at a fixed rate $1-\eta=10\%$, which is currently the default in their published code (the choice of symbol and parameterization for $\eta$ is to allow direct comparison with their paper). Obviously, for our application, a fixed value of $\eta$ would result in an inconsistent estimate. 

\begin{figure}
\centering
\includegraphics[width=8cm]{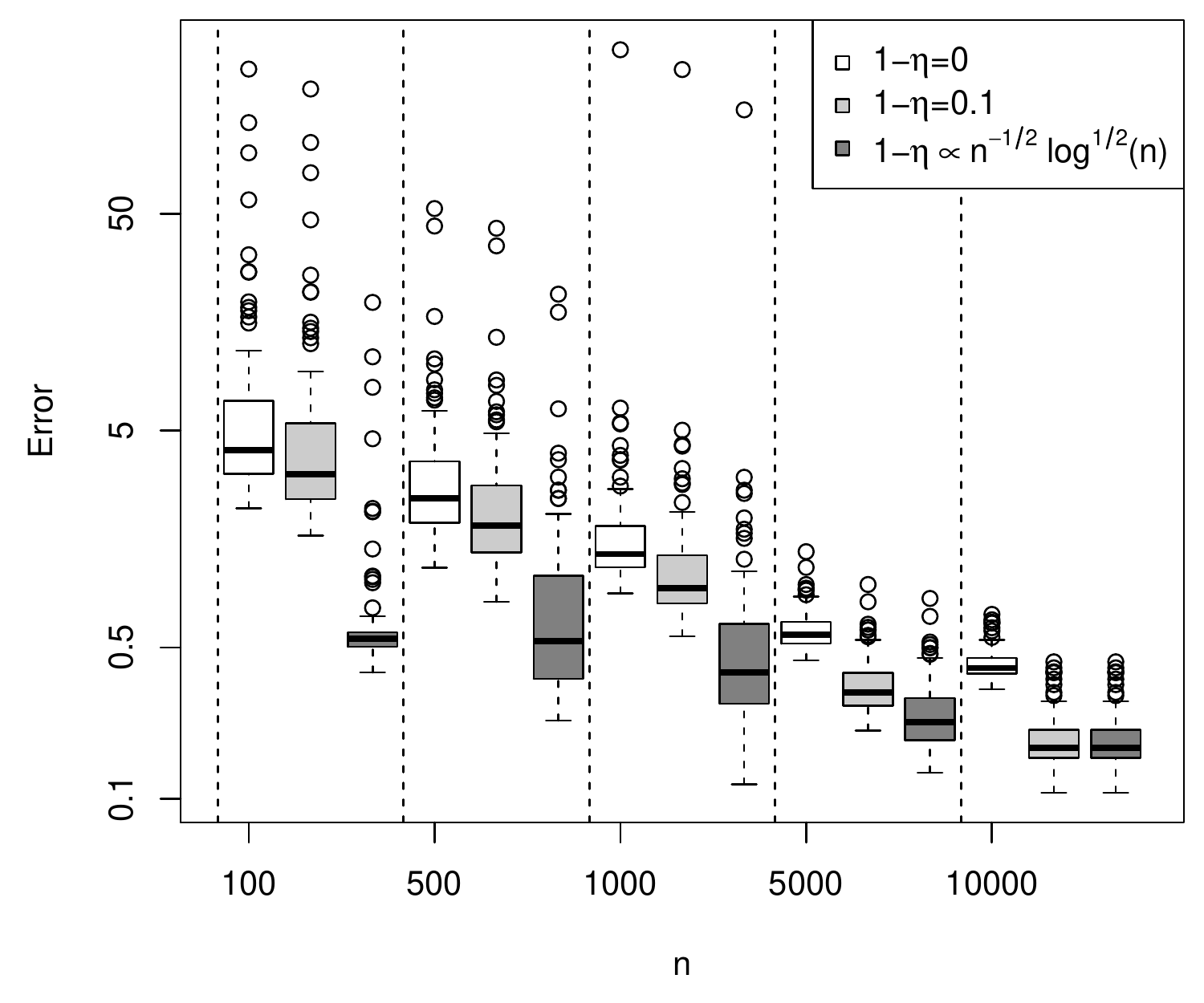}
\caption{Boxplot of the estimation error as $n$ increases for different estimates of $\hat B$, where $B$ is $3 \times 3$ with full rank (given explicitly in main text) and $\alpha = (1,1,1)$ based on 100 simulations for each $n$. The first estimate uses the MVECP-$k$, i.e. $\eta=1$ (in white), the second uses the polytope shrunk by 10\%, i.e. $\eta = 0.9$ (light grey), and the third uses a decreasing shrinking rate proportional to $n^{-1/2}\log^{1/2}(n)$ instead (dark grey). The error is computed as the minimum distance between $\hat B$ and $B$ in Frobenius norm, minimizing over all row-column permutations of $\hat B$. The y-axis is in the log-scale.}\label{eq:boxplot_error}
\end{figure}

Figure~\ref{eq:boxplot_error} shows the estimation improvements that result from shrinking the polytope at the rate derived above. The constant was chosen such that $1-a n^{-1/2}\log^{1/2}(n) = 0.9$  at $n=10,000$, so that shrinking at a fixed or varying $\eta$ results in identical performance at $n=10,000$. Performance is otherwise always superior using a varying $\eta$, rather than fixing $\eta = 0.9$ (default shrinking) or $\eta=1$ (no shrinking). 

\section{Conclusion}\label{sec:conclusion}
This paper presents a simple estimation procedure for the undirected mixed membership stochastic blockmodel, based on its connection to the random dot product graph, using adjacency spectral embedding. We prove consistency under reasonable conditions, and propose a bias correction for finite samples exploiting, respectively, recent $2 \rightarrow \infty$ norm and central limit theorems for random dot product graphs.

\bibliographystyle{apalike}
\bibliography{graph_embedding}

\begin{thebibliography}{}

\bibitem[Addario-Berry et~al., 2015]{addario2015exceptional}
Addario-Berry, L., Bhamidi, S., Bubeck, S., Devroye, L., Lugosi, G., and
  Oliveira, R.~I. (2015).
\newblock Exceptional rotations of random graphs: a vc theory.
\newblock {\em Journal of Machine Learning Research}, 16:1893--1922.

\bibitem[Airoldi et~al., 2008]{airoldi2008mixed}
Airoldi, E.~M., Blei, D.~M., Fienberg, S.~E., and Xing, E.~P. (2008).
\newblock Mixed membership stochastic blockmodels.
\newblock {\em Journal of Machine Learning Research}, 9(Sep):1981--2014.

\bibitem[Anandkumar et~al., 2014]{anandkumar2014tensor}
Anandkumar, A., Ge, R., Hsu, D.~J., and Kakade, S.~M. (2014).
\newblock A tensor approach to learning mixed membership community models.
\newblock {\em Journal of Machine Learning Research}, 15(1):2239--2312.

\bibitem[Athreya et~al., 2016]{athreya2016limit}
Athreya, A., Priebe, C.~E., Tang, M., Lyzinski, V., Marchette, D.~J., and
  Sussman, D.~L. (2016).
\newblock A limit theorem for scaled eigenvectors of random dot product graphs.
\newblock {\em Sankhya A}, 78(1):1--18.

\bibitem[Barab{\'a}si, 2016]{barabasi2016network}
Barab{\'a}si, A.-L. (2016).
\newblock {\em Network science}.
\newblock Cambridge University Press.

\bibitem[Caron and Fox, 2017]{caron2014sparse}
Caron, F. and Fox, E.~B. (2017).
\newblock Sparse graphs using exchangeable random measures.
\newblock {\em Journal of the Royal Statistical Society: Series B,}, 79:1--44.

\bibitem[Choi et~al., 2012]{choi2012stochastic}
Choi, D.~S., Wolfe, P.~J., and Airoldi, E.~M. (2012).
\newblock Stochastic blockmodels with a growing number of classes.
\newblock {\em Biometrika}, page asr053.

\bibitem[Embrechts et~al., 2013]{embrechts2013modelling}
Embrechts, P., Kl{\"u}ppelberg, C., and Mikosch, T. (2013).
\newblock {\em Modelling extremal events: for insurance and finance},
  volume~33.
\newblock Springer Science \& Business Media.

\bibitem[Gao et~al., 2015]{gao2015rate}
Gao, C., Lu, Y., Zhou, H.~H., et~al. (2015).
\newblock Rate-optimal graphon estimation.
\newblock {\em The Annals of Statistics}, 43(6):2624--2652.

\bibitem[Gopalan and Blei, 2013]{gopalan2013efficient}
Gopalan, P.~K. and Blei, D.~M. (2013).
\newblock Efficient discovery of overlapping communities in massive networks.
\newblock {\em Proceedings of the National Academy of Sciences},
  110(36):14534--14539.

\bibitem[Ho et~al., 2016]{ho2016latent}
Ho, Q., Yin, J., and Xing, E.~P. (2016).
\newblock Latent space inference of internet-scale networks.
\newblock {\em Journal of Machine Learning Research}, 17(78):1--41.

\bibitem[Holland et~al., 1983]{holland1983stochastic}
Holland, P.~W., Laskey, K.~B., and Leinhardt, S. (1983).
\newblock Stochastic blockmodels: First steps.
\newblock {\em Social networks}, 5(2):109--137.

\bibitem[Karrer and Newman, 2011]{karrer2011stochastic}
Karrer, B. and Newman, M.~E. (2011).
\newblock Stochastic blockmodels and community structure in networks.
\newblock {\em Physical Review E}, 83(1):016107.

\bibitem[Klopp et~al., 2017]{klopp2017oracle}
Klopp, O., Tsybakov, A.~B., Verzelen, N., et~al. (2017).
\newblock Oracle inequalities for network models and sparse graphon estimation.
\newblock {\em The Annals of Statistics}, 45(1):316--354.

\bibitem[Lei et~al., 2016]{lei2016goodness}
Lei, J. et~al. (2016).
\newblock A goodness-of-fit test for stochastic block models.
\newblock {\em The Annals of Statistics}, 44(1):401--424.

\bibitem[Lei and Rinaldo, 2015]{lei2015}
Lei, J. and Rinaldo, A. (2015).
\newblock Consistency of spectral clustering in stochastic block models.
\newblock {\em Ann. Statist.}, 43(1):215--237.

\bibitem[Lin et~al., 2016]{lin2016fast}
Lin, C.-H., Chi, C.-Y., Wang, Y.-H., and Chan, T.-H. (2016).
\newblock A fast hyperplane-based minimum-volume enclosing simplex algorithm
  for blind hyperspectral unmixing.
\newblock {\em IEEE Transactions on Signal Processing}, 64(8):1946--1961.

\bibitem[Lyzinski et~al., 2017]{lyzinski2017community}
Lyzinski, V., Tang, M., Athreya, A., Park, Y., and Priebe, C.~E. (2017).
\newblock Community detection and classification in hierarchical stochastic
  blockmodels.
\newblock {\em IEEE Transactions on Network Science and Engineering},
  4(1):13--26.

\bibitem[Mao et~al., 2016]{mao2016provable}
Mao, X., Sarkar, P., and Chakrabarti, D. (2016).
\newblock Provable symmetric nonnegative matrix factorization for overlapping
  clustering.
\newblock {\em arXiv preprint arXiv:1607.00084}.

\bibitem[Nickel, 2006]{nickel06}
Nickel, C. (2006).
\newblock {\em Random Dot Product Graphs: A Model for Social Networks}.
\newblock PhD thesis, Johns Hopkins University.

\bibitem[Packer, 2002]{packer2002np}
Packer, A. (2002).
\newblock {NP}-hardness of largest contained and smallest containing simplices
  for v-and h-polytopes.
\newblock {\em Discrete and Computational Geometry}, 28(3):349--377.

\bibitem[Rohe et~al., 2011]{rohe2011spectral}
Rohe, K., Chatterjee, S., and Yu, B. (2011).
\newblock Spectral clustering and the high-dimensional stochastic blockmodel.
\newblock {\em The Annals of Statistics}, pages 1878--1915.

\bibitem[Von~Luxburg, 2007]{von2007tutorial}
Von~Luxburg, U. (2007).
\newblock A tutorial on spectral clustering.
\newblock {\em Statistics and computing}, 17(4):395--416.

\bibitem[Wang et~al., 2016]{wang2016trend}
Wang, Y.-X., Sharpnack, J., Smola, A., and Tibshirani, R.~J. (2016).
\newblock Trend filtering on graphs.
\newblock {\em Journal of Machine Learning Research}, 17(105):1--41.

\bibitem[Young and Scheinerman, 2007]{young2007random}
Young, S.~J. and Scheinerman, E.~R. (2007).
\newblock Random dot product graph models for social networks.
\newblock In {\em International Workshop on Algorithms and Models for the
  Web-Graph}, pages 138--149. Springer.

\bibitem[Yu et~al., 2015]{yu2015useful}
Yu, Y., Wang, T., and Samworth, R.~J. (2015).
\newblock A useful variant of the {D}avis--{K}ahan theorem for statisticians.
\newblock {\em Biometrika}, 102(2):315--323.

\bibitem[Zhang et~al., 2014]{zhang2014detecting}
Zhang, Y., Levina, E., and Zhu, J. (2014).
\newblock Detecting overlapping communities in networks using spectral methods.
\newblock {\em arXiv preprint arXiv:1412.3432}.

\end{thebibliography}
\appendix
\section{Appendix}
\begin{proof}[Proof of Theorem~\ref{thm:main_theorem}]
The sequence $d_n$ and event $E_n$ defined at the end of Section~\ref{sec:spectral_estimation} are now denoted $d_n^{(1)}$ and $E_n^{(1)}$. Let $Y = (X - \mu) U_{\Gamma} = [Y_1, \ldots, Y_n]^T$ where $\mu = [\E(X_1), \ldots, \E(X_1)]^T \in \R^{n \times d}$, $\Gamma = \E[\{X_1 - \E(X_1)\}\{X_1 - \E(X_1)\}^T]$, and $U_{\Gamma}$ is a matrix containing $d-1$ orthonormal eigenvectors corresponding to the $d-1$ largest eigenvalues of $\Gamma$. 

Since $X_i$ have bounded support, the standard mean and covariance estimates are consistent, so that there exists a sequence $d^{(2)}_n \rightarrow 0$ such that for the event $E^{(2)}_n$: 
\[
\left\lVert n^{-1} \sum_{i=1}^n X_{i} - \E(X_1) \right\rVert \leq d^{(2)}_n \: \text{and} \:  \lVert C - \Gamma \rVert_{\text{F}}  \leq d^{(2)}_n,
\]
where  $C = n^{-1}(X-M)^T(X-M)$ and $M = [n^{-1} \sum_{i=1}^n X_{i}, \ldots, n^{-1} \sum_{i=1}^n X_{i}]^T \in \R^{n \times d}$, we have $\Prob\left(E^{(2)}_n\right) \rightarrow 1$. Then $E^{(1)}_n$  and  $E^{(2)}_n$ together imply that
\begin{equation}
\left\lVert n^{-1} \sum_{i=1}^n W\hat X_{i} -  \E(X_1) \right\rVert \leq d^{(3)}_n \: \text{and} \:  \lVert W \hat C W^T - \Gamma \rVert_{\text{F}}  \leq d^{(3)}_n, \label{eq:mean_and_covariance_bound}
\end{equation}
for some sequence $d^{(3)}_n \rightarrow 0$. Since $\mathcal S$ is $(d-1)$-dimensional, $\Gamma$ has $d-1$ positive and one zero eigenvalue. Therefore, under $E^{(1)}_n$  and  $E^{(2)}_n$, by a variant of the Davis-Kahan theorem \citep{yu2015useful}, there exists an orthogonal matrix $O \in \mathrm{O}(d-1)$ such that
 \begin{equation}
\lVert W U_{\hat C} O^T - U_{\Gamma}\rVert_{\text{F}} \leq 2^{3/2} \lambda_{d-1}^{-1} d^{(3)}_n = d^{(4)}_n \rightarrow 0, \label{eq:UhatCbound}
\end{equation} 
where $U_{\hat C}$ is a matrix containing $d-1$ orthonormal eigenvectors corresponding to the $d-1$ largest eigenvalues of $\hat C$, and $\lambda_{d-1}$ is the smallest non-zero eigenvalue of $\Gamma$. Recall from Definition~\ref{def:spectral_estimation} that $\tilde X_1, \ldots, \tilde X_n$ are the PCA projections of $\hat X_1, \ldots, \hat X_n$ respectively. We have,
$(O \tilde X_i - Y_i) = V^{(1)}_{i} + V^{(2)}_{i}$
where
\begin{align*}
V^{(1)}_{i}(n) &= O U_{\hat C}^T\left\{(\hat X_i - W^T X_i) + \left(W^T \E(X_1) - n^{-1} \sum_{i=1}^n \hat X_{i}\right)\right\},\\
V^{(2)}_{i}(n) &= \left(O U_{\hat C}^T W^T - U_{\Gamma}^T\right)\{X_i - \E(X_1)\}.
\end{align*}
Under $E^{(1)}_n$  and  $E^{(2)}_n$, we have $\lVert V^{(1)}_{i} \rVert \leq d^{(2)}_n + d^{(3)}_n$, using $\lVert B x \rVert_{2} \leq \lVert B \rVert_{2} \lVert x \rVert$ where $\lVert \cdot \rVert_{2}$ denotes the spectral norm of a matrix and noting that $\lVert O U_{\hat C}^T \rVert_{2} = 1$, whereas $\lVert V^{(2)}_{i} \rVert \leq d^{(4)}_n c^*$, where $c^*$ is the maximum possible norm of $X_i - \E(X_1)$ (using the Cauchy-Schwarz inequality for the Frobenius norm), and therefore $\lVert O \tilde X_i - Y_i\rVert \leq d^{(2)}_n + d^{(3)}_n + d^{(4)}_n c^* = d^{(5)}_n \rightarrow 0$.

Let $\mathcal{R} = \{U_\Gamma^T [x - \E(X_1)]: x \in \mathcal{S} \}$ denote the support of $Y_i$, a $(d-1)$-dimensional convex $l$-polytope with vertices $w_1, \ldots, w_l$. Under $E^{(1)}_n$  and  $E^{(2)}_n$, $O \tilde X_i$ are enclosed in a convex $l$-polytope formed by moving each of the $(d-2)$-dimensional hyperplanes containing a side of $\mathcal{R}$ by $d^{(5)}_n$, away from the polytope centre, and parallel to the original plane.
Under $E^{(1)}_n$  and  $E^{(2)}_n$, the MVECP-$k$ necessarily has a volume at least as small, so that $\lambda(O \tilde{\mathcal{S}}) \leq \lambda(\mathcal{R}) + \delta^{(1)}_n$, where $\lambda$ denotes Lebesgue measure on $\R^{d-1}$, and $\delta^{(1)}_n$ is some sequence that can be chosen so that $\delta^{(1)}_n \rightarrow 0$. 


Since each $Y_i$ has a positive probability of falling within any neighbourhood of $w_j, j = 1, \ldots, l$, there also exists a function $d^{(6)}_n \rightarrow 0$ such that, for the event $E^{(3)}_n$: 
\begin{equation*}
\max_{j\in\{1, \ldots, l\}}\left(\min_{i \in \{1, \ldots, n\}} \lVert Y_i-w_j \rVert \right) \leq d^{(6)}_n, 
\end{equation*}
we have $\Prob\left(E^{(3)}_n\right) \rightarrow 1$. Let $\mathcal H$ denote the convex hull of $H_1, \ldots, H_l$, where
 \[H_j = \underset{x \in \{O \tilde X_i\}}{\argmin} \lVert x - w_j \rVert, \quad j = 1, \ldots, l.\]
Then  $E^{(1)}_n$, $E^{(2)}_n$ and $E^{(3)}_n$ together imply that $\lambda(\mathcal{R} \ominus \mathcal{H}) \leq \delta^{(2)}_n$, where $\delta^{(2)}_n$ is some sequence that can be chosen so that $\delta^{(2)}_n \rightarrow 0$. 

Therefore,   $E^{(1)}_n$, $E^{(2)}_n$ and $E^{(3)}_n$ together imply
\begin{align*}
\lambda(O \tilde{\mathcal{S}} \cup \mathcal{R}) & \leq \lambda(O \tilde{\mathcal{S}} \cup \mathcal{H}) + \delta^{(2)}_n = \lambda(O \tilde{\mathcal{S}}) + \delta^{(2)}_n,\\
\lambda(O \tilde{\mathcal{S}} \cap \mathcal{R}) &\geq \lambda(\mathcal{H} \cap \mathcal{R}) \geq \lambda(\mathcal{R}) - \delta^{(2)}_n,
\end{align*}
using direct set algebra, and the fact that $\mathcal{H} \subseteq O\tilde{\mathcal{S}}$. Hence, 
\begin{align*}
\lambda(O \tilde{\mathcal{S}} \ominus \mathcal{R}) &= \lambda(O \tilde{\mathcal{S}} \cup \mathcal{R}) - \lambda(O \tilde{\mathcal{S}} \cap \mathcal{R})\\
&\leq \lambda(O \tilde{\mathcal{S}}) - \lambda(\mathcal{R}) + 2 \delta^{(2)}_n \leq \delta^{(1)}_n+ 2\delta^{(2)}_n
.\end{align*}
For any $\delta,\epsilon>0$, let $n_1$ be such that $\Prob\left(E^{(i)}_n\right) \geq 1-\epsilon/3$, $i = 1,2,3$, for all $n \geq n_1$. Let $n_2$ be such that $\delta^{(1)}_n+ 2\delta^{(2)}_n \leq \delta$, for all $n \geq n_2$. Then, for any $n \geq \max(n_1, n_2)$, 
\begin{equation}
\Prob\{\lambda(O \tilde{\mathcal{S}} \ominus \mathcal{R}) > \delta\} \leq \Prob\left(\bar E^{(1)}_n \cup \bar E^{(2)}_n \cup \bar E^{(3)}_n\right) \leq \epsilon, \label{eq:pca_formula}
\end{equation}
using Boole's inequality. Therefore $\tilde{\mathcal{S}}$ is a consistent estimate of $\mathcal{R}$ (up to orthogonal transformation). It is not hard to show consistency is preserved after PCA reconstruction, using $\mathcal{S} = \{U_{\Gamma} x + \E(X_1): x \in \mathcal{R}\}$ (recall $\hat{\mathcal{S}} = \{U_{\hat C} x + n^{-1} \sum_{i=1}^n \hat X_{i}: x \in \tilde{\mathcal{S}}\}$) and equations (\ref{eq:mean_and_covariance_bound}-\ref{eq:UhatCbound}).

Now, assume $k = l$, so that $\hat{\mathcal{S}}$ and $\mathcal S$ are two $(d-1)$-dimensional convex $k$-polytopes. Suppose there exists $W \in \text{O}(d)$ such that $\lambda(W \hat{\mathcal{S}} \ominus \mathcal S) \leq \delta$ for some $\delta$, and consider the vertex estimation error $\max_{i\in \{1, \ldots, k\}} \lVert W \hat V_i - v_i \rVert$, choosing the order of the vertices to achieve minimum error. By an argument shown in Figure~\ref{fig:vtx_convergence_proof}, we must have $\max_{i=1, \ldots, n} \lVert W \hat V_i - v_i \rVert \leq \Delta$, where $\Delta$ is a positive number that can be made arbitrarily small by reducing $\delta$ with $\mathcal S$ fixed. Therefore, for any $\Delta > 0$,
\begin{equation}
\underset{n \rightarrow \infty} {\lim} \Prob\left\{\underset{W \in \mathrm{O}(d)}{\min}\left(\max_{i\in \{1, \ldots, k\}} \lVert W \hat V_i - v_i \rVert\right)  \leq \Delta\right\} = 1 \label{eq:consistency_vtx},
\end{equation}
implying that $\hat B$ is a consistent estimator of $B$ up to row-column permutations.

Finally, assume $d = k = l$. As we consistently estimate both $v_1, \ldots, v_k$, see Equation \eqref{eq:consistency_vtx}, and $X_i$, see Equation \eqref{eq:bound}, then $\hat{\pi}_1, \ldots, \hat{\pi}_n$ are collectively consistent estimates of $\pi_1, \ldots, \pi_n$ respectively, proving the LHS of \eqref{eq:pi_estimate}. This in turns implies that $\hat \alpha$ is a consistent estimate of $\alpha$, by the consistency of the maximum likelihood estimator, proving the RHS of \eqref{eq:pi_estimate}.
\end{proof}

\begin{figure}
\centering
\includegraphics[width=8cm]{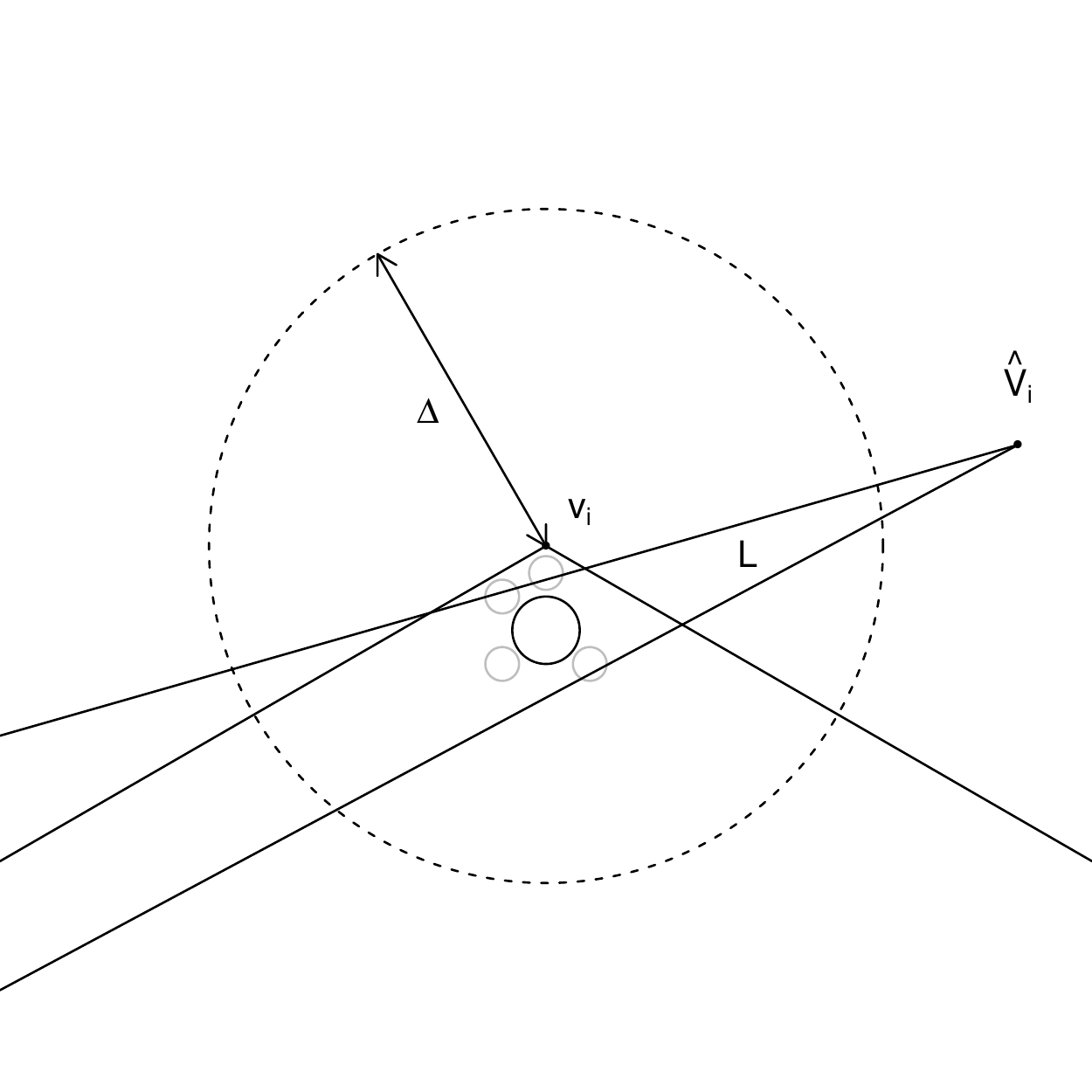}
\caption{Illustration of the proof of the convergence of $\hat V_i$ to $v_i$ given the convergence of $\hat{\mathcal{S}}$ to $\mathcal{S}$ in symmetric difference. If $\lVert W \hat V_i - v_i \rVert > \Delta$, after minimizing over vertex permutations, a ball of radius $\min(\Delta, \lVert v_i - v_{j\neq i} \rVert)$ (dotted line) can be drawn around $v_i$ so that neither $W \hat{\mathcal{S}}$ nor $\mathcal{S}$ have (other) vertices within the ball. The polytope $W \hat{\mathcal{S}}$ must include the whole black ball if it does not exclude any of the grey balls, in which case it must include a portion $L$ of the dotted ball that is not in $\mathcal{S}$. The solid black and grey balls can be chosen, based only on $\Delta$ and $\mathcal{S}$, so that they and $L$ have a volume exceeding $\delta >0$, which is a function of $\Delta$ and $\mathcal{S}$ only. Therefore $\lVert W \hat V_i - v_i \rVert > \Delta$ implies that $W \hat{\mathcal{S}}$ either excludes a grey ball or includes $L$, either of which imply that $W \hat{\mathcal{S}} \ominus \mathcal S > \delta$, so must occur with vanishing probability. See proof of Theorem \ref{thm:main_theorem}.}
\label{fig:vtx_convergence_proof}
\end{figure}

\end{document}